\documentclass[ USenglish, cleveref, cref, thm-restate]{article}
\usepackage{mathtools}
\usepackage{float}
\usepackage{physics}
\usepackage{tikz}

\usetikzlibrary{decorations.pathmorphing}
\usetikzlibrary{decorations.pathreplacing}
\usetikzlibrary{calc}
\usetikzlibrary{matrix}

\usepackage{hyperref}
\usepackage{cleveref}

\usepackage{svg}
\usepackage{amsmath,amsthm}
\usepackage{thmtools,mathtools}
\usepackage{amsfonts}
\usepackage{bookmark}
\usepackage[algo2e,linesnumbered,boxed,ruled,vlined]{algorithm2e}

\newtheorem{theorem}{Theorem}
\newtheorem{definition}{Definition}
\newtheorem{lemma}{Lemma}
\newtheorem{remark}{Remark}
\newtheorem{corollary}{Corollary}

\newcommand{\wtilde}[1]{\widetilde{#1}}
\newcommand{\wt}[1]{\widetilde{#1}}

\newcommand{\gma}{\gamma}

\newcommand{\Omg}{\Omega}
\newcommand{\dta}{\delta}

\newcommand{\eps}{\epsilon}

\newcommand{\cA}{\mathcal{A}}
\newcommand{\cB}{\mathcal{B}}

\newcommand{\cO}{\mathcal{O}}

\newcommand{\Od}{\cO}
\newcommand{\tO}{\tilde{\Od{}}}
\newcommand{\tOmg}{\tilde{\Omg{}}}

\newcommand{\ts}{\tilde{s}}
\newcommand{\tS}{\tilde{S}}
\newcommand{\tA}{\tilde{A}}
\newcommand{\tB}{\tilde{B}}
\newcommand{\TS}{\tilde{S}}%
\newcommand{\tn}{\tilde{n}}
\newcommand{\ti}{\tilde{i}}

\newcommand{\td}{\tilde{d}}
\newcommand{\Tt}{\tilde{t}}%

\newcommand{\PARITY}{\mathsf{PARITY}}
\newcommand{\LCS}{\mathsf{LCS}}
\newcommand{\LCSRLE}{\mathsf{LCS\text{-}RLE}}
\newcommand{\LCSRLEP}{\LCSRLE^\mathsf{p}}
\newcommand{\ELLCSRLE}{\mathsf{EL\text{-}LCS\text{-}RLE}}
\newcommand{\DLLCSRLE}{\mathsf{DL\text{-}LCS\text{-}RLE}}

\newcommand{\calA}{\mathcal{A}}

\newcommand{\Piv}{P^{\text{ -1}}}

\newcommand{\polylog}{\mathrm{polylog}}
\newcommand{\poly}{\mathrm{poly}}

\DeclarePairedDelimiter\set{\{}{\}}

\newcommand{\PH}{{}\cdot{}}

\newcommand{\Ie}{I.e.\ }
\newcommand{\ie}{i.e.\ }
\newcommand{\eg}{e.g.,\ }

\tikzset {
	wavy/.style={ decoration={coil,aspect=0,amplitude=1pt} },
	desc/.style={ node font=\footnotesize },
	brkt/.style={ pos=0.25 },
}

\ExplSyntaxOn
\NewDocumentCommand{\rle}{ m } {
	\group_begin:
	\bool_set_false:N \l_tmpa_bool
	\clist_map_inline:nn {#1} {
		\bool_set_inverse:N \l_tmpa_bool
		\bool_if:NTF \l_tmpa_bool	{
			\tl_set:Nn \l_tmpa_tl {##1}
		} {
			\texttt{\l_tmpa_tl}^{##1}
		}
	}
	\group_end:
}

\NewDocumentCommand{\torle}{ m } {
	\group_begin:
	\str_clear:N \l_tmpa_str
		\str_map_inline:nn {#1} {
			\str_if_in:NnTF \l_tmpa_str {##1} {
				\str_put_right:Nn \l_tmpa_str {##1}
			} {
				\__rle_str_run:N \l_tmpa_str
				\str_set:Nn \l_tmpa_str {##1}
			}
		}
	\__rle_str_run:N \l_tmpa_str
	\group_end:
}

\cs_new:Nn \__rle_str_run:N {
	\str_if_empty:NF #1 {
		\texttt{\str_head:N #1}^{\str_count:N #1}
	}
}

\cs_new:Nn \__rle_parse:nNN {
	\str_map_inline:nn {#1} {
		\str_if_in:nnTF {1234567890} {##1} {
			\seq_put_right:Nn #3 {##1}
		} {
			\seq_put_right:Nn #2 {##1}
		}
	}
}

\cs_new:Nn \__rle_plus_fold:NNnn {
	\seq_map_indexed_inline:Nn #2 {
		\int_compare:nNnT {##1} > {#4} {
			\seq_map_break:
		}
		\int_compare:nNnF {##1} < {#3} {
			\int_add:Nn #1 {##2}
		}
	}
}

\cs_new:Nn \__rle_path:nn {
	\str_if_eq:nnTF {#1} {notail} {
		\path[save~path=\rlepath] (0, 0) rectangle (#2, \h);
	} {
		\path[save~path=\rlepath]
			(#2-\h/8+2, 0) decorate[wavy] {-- ++(0, \h) }
			-- (0, \h) |- cycle
			(#2+\h/8+2, 0) decorate[wavy] {-- ++(0, \h) }
			-- ++(2-\h/8, 0) |- cycle;
	}
}

\cs_new:Nn \__rle_run_text:nn {
	\node at (\l_tmpa_int+#2*0.5,\h/2) {\texttt{#1}$^{#2}$};
	\int_add:Nn \l_tmpa_int {#2}
	\draw (\l_tmpa_int,0) -- ++(0, \h);
}

\int_new:N \l_a_len_int
\int_new:N \l_h_len_int
\int_new:N \l_t_len_int
\int_new:N \l_r_len_int
\int_new:N \l_b_len_int
\int_new:N \l_r_cnt_int

\NewDocumentCommand{\drawrle}{ O{tail} m O{1} O{10000} o } {
	\group_begin:
	\__rle_parse:nNN {#2}
		\l_tmpa_seq \l_tmpb_seq

	\int_set:Nn \l_r_cnt_int {\seq_count:N \l_tmpb_seq}

	\__rle_plus_fold:NNnn
		\l_a_len_int \l_tmpb_seq {1} {\l_r_cnt_int}

	\__rle_plus_fold:NNnn
		\l_h_len_int \l_tmpb_seq {1} {#3}

	\__rle_plus_fold:NNnn
		\l_t_len_int \l_tmpb_seq {#3+1} {\l_r_cnt_int}

	\__rle_plus_fold:NNnn
		\l_r_len_int \l_tmpb_seq {#3} {#3}
	\int_decr:N \l_r_len_int

  \str_if_eq:nnTF {#5} {-NoValue-} {
		\__rle_plus_fold:NNnn
			\l_b_len_int \l_tmpb_seq {#3+1} {\int_min:nn {#4} {\l_r_cnt_int}}
		\int_incr:N \l_b_len_int
		\int_compare:nNnT {#4} > {\l_r_cnt_int} {
			\str_if_eq:nnF {#1} {notail} {
				\int_add:Nn \l_b_len_int {4}
			}
		}
	} {
		\int_set:Nn \l_b_len_int {#5-\l_r_len_int-1}
		\int_compare:nNnTF {\l_b_len_int+\l_h_len_int} > {\l_a_len_int} {
				\int_add:Nn \l_b_len_int {4}
		} {
				\int_add:Nn \l_b_len_int {1}
		}
	}

	\begin{scope}[xscale=\xs, shift={(-\l_h_len_int, 0)}]

		\__rle_path:nn {#1} {\l_a_len_int}

		\coordinate (HEAD) at (0, \h/2);
		\coordinate (R) at (\l_h_len_int-1-\l_r_len_int, 0);
		\coordinate (B) at (\l_h_len_int-1+\l_b_len_int, 0);

		\begin{scope}
			\clip[use~path=\rlepath];
			\fill[fill=red!30]  (\l_h_len_int-1, 0) rectangle ++(-\l_r_len_int, \h);
			\fill[fill=cyan!30] (\l_h_len_int-1, 0) rectangle ++( \l_b_len_int, \h);
		\end{scope}

		\draw[dashed] (\l_h_len_int-1, 0) -- ++(0, \h);
		\draw[thick, use~path=\rlepath];

		\int_set:Nn \l_tmpa_int {0}
		\int_step_inline:nnn {1} {\l_r_cnt_int} {
			\__rle_run_text:nn
				{\seq_item:Nn \l_tmpa_seq {##1}}
				{\seq_item:Nn \l_tmpb_seq {##1}}
		}
		\str_if_eq:nnF {#1} {notail} {
			\node at (\l_tmpa_int+2/2, \h/2) {$\cdots$};
			\node at (\l_tmpa_int+2+2/2, \h/2) {$\cdots$};
		}
	\end{scope}
	\group_end:
}
\ExplSyntaxOff

\newcommand{\ldcp}{\mathrm{ldcp}}
\newcommand{\pos}{\mathrm{pos}}

\begin{document}
\title{Near-Optimal Quantum Algorithm for Finding the Longest Common Substring between Run-Length Encoded Strings}
\author{
    Tzu-Ching Lee \thanks{National Tsing Hua University, Taiwan.} \and
	Han-Hsuan Lin\thanks{National Tsing Hua University, Taiwan. Supported by NSTC QC project under Grant no.\  111-2119-M-001-006- and 110-2222-E-007-002-MY3.} 
}

\maketitle

\begin{abstract}
	We give a near-optimal quantum algorithm for the longest common substring (LCS) problem between two run-length encoded (RLE) strings, with the assumption that the prefix-sums of the run-lengths are given.
	Our algorithm costs $\tO(n^{2/3}/d^{1/6-o(1)}\cdot\polylog(\tn))$ time, while the query lower bound for the problem is $\tOmg(n^{2/3}/d^{1/6})$, where $n$ and $\tn$ are the encoded and decoded length of the inputs, respectively, and $d$ is the encoded length of the LCS.
    We justify the use of prefix-sum oracles for two reasons. First, we note that creating the prefix-sum oracle only incurs a constant overhead in the RLE compression. Second, we show that, without the oracles, there is a $\Omg(n/\log^2n)$ lower bound on the quantum query complexity of finding the LCS given two RLE strings due to a reduction of $\PARITY$ to the problem.
	With a small modification, our algorithm also solves the longest repeated substring problem for an RLE string.
\end{abstract}

\section{Introduction}

String processing is an important field of research in theoretical computer science.
There are many results for various classic string processing problems, such as string matching \cite{KnuthMP77, BoyerM77, KarpR87}, longest common substring, and edit distance.
The development of string processing algorithms has led to the discovery of many impactful computer science concepts and tools, including dynamic programming, suffix tree~\cite{Weiner73, Farach97} and trie \cite{Fredkin60}.
String processing also has applications in various fields such as bioinformatics \cite{NeedlemanW70}, image analysis \cite{HindsFD90}, and compression \cite{ZivL77}.

A natural extension of string processing is to do it between \emph{compressed strings}. Ideally, the time cost of string processing between compressed strings would be independent of the decoded lengths of the strings. Since the compressed string can be much shorter than the original string, this would significantly save computation time. Whether such fast string processing is possible depends on what kind of compression scheme we are using.

Run-Length Encoding (RLE) is a simple way to compress strings. In RLE, the consecutive repetition of a character (run) is replaced by a character-length pair, the character itself and the length of run.
For example, the RLE of the string \texttt{aaabcccdd} is $\rle{a,3,b,1,c,3,d,2}$.
RLE is a common method to compress fax data \cite{ITU-T.4:2004}, and is also part of the JPEG and TIFF image standard \cite{ISO10918:1994, ISO12639:1998}.
String processing on RLE strings has been extensively studied. Apostolico, Landau, and Skiena gave an algorithm in time $\Od(n^2\log n)$ to find the longest common sequence between two RLE compressed strings \cite{ApostolicoLS99}, where $n$ is the length of the compressed strings.
Hooshmand, Tavakoli, Abedin, and Thankachan obtained an $\Od(n\log n)$-time algorithm on computing the Average Common Substring with RLE inputs \cite{HooshmandTAT18}.
Chen and Chao proposed an algorithm to compute the edit distance between two RLE strings \cite{ChenC13}, and Clifford, Gawrychowski, Kociumaka, Martin and Uznanski further improved the result to near-optimal in \cite{CliffordGKMU19}, which runs in $\Od(n^2\log n)$ time.

Alternatively, another way to speed up string processing is to use \emph{quantum algorithms}. If we use quantum algorithms for string processing, it is possible to get the time cost \emph{sublinear} in the input length because the quantum computer can read the strings in superposition. One of the earliest such results was by Hariharan and Vinay \cite{HariharanV03}, who constructed a $\tO(\sqrt{n})$-time string matching quantum algorithm, in which Grover's search \cite{Grover96} and Vishkin's deterministic sampling technique \cite{Vishkin90} were used to reach this near-optimal time complexity.
Le Gall and Seddighin \cite{LeGallS22} use quantum walk \cite{MagniezNRS11} to obtain several sublinear-time quantum algorithms for various string problems, including a $\tO(n^{5/6})$-time algorithm for longest common substring (LCS), a $\tO(\sqrt{n})$-time algorithm for longest palindrome substring (LPS), and a $\tO(\sqrt{n})$-time algorithm for approximating the Ulam distance.
Another work is done by Akmal and Jin \cite{AkmalJ22}, using string synchronizing sets \cite{KempaK19} with quantum walk \cite{MagniezNRS11}, showing that LCS can be solved in $\tO(n^{2/3})$ quantum time. They also introduced a $n^{1/2+o(1)}$ algorithm for the lexicographically minimal string rotation problem, and a $\tO(\sqrt{n})$-time algorithm for longest square substring problem in the same paper.
\cite{jin2023quantum} further improves on \cite{AkmalJ22} with a better quantum string synchronizing set construction, getting $\tO(n^{2/3}/d^{1/6-o(1)})$ quantum time on the LCS problem, which is near-optimal with respect to both $n$ and $d$, the length of the common substring. \cite{jin2023quantum} also gives a $\tO(kn^{1/2})$-time quantum algorithm for the $k$-mismatch Matching problem.

In this work, we combine the two above ideas and investigate the possibility of using \emph{quantum algorithm} to do string processing on \emph{compressed strings}, while keeping the advantages of both methods. Thus, we ask the following question:
\begin{center}
Is it possible to have a quantum string processing algorithm on compressed strings whose time cost is sublinear in the encoded lengths of the strings and independent of the decoded lengths?\footnote{With a non-trivial string problem and a non-trivial compression scheme.}
\end{center}

  The main contribution of this paper is the first almost\footnote{We have $\polylog$ dependence on the decoded length $\tn$.} affirmative answer to the above question, an almost optimal quantum algorithm computing the longest common substring (LCS) between two RLE strings:

\begin{theorem}[Informal]
There is a quantum algorithm that finds the RLE of an LCS given two RLE strings in $\Od(n^{2/3}/d^{1/6-o(1)}\cdot\polylog(n)\cdot\polylog(\tn))$ time, with oracle access to the RLE strings and the prefix-sum\footnote{defined in \cref{dfn:prefix-sum}} of their runs, where $n$ and $\tn$ are the encoded length and the decoded length of the inputs, respectively, and $d$ is the encoded length of the longest common substring.
\end{theorem}

Note that we modify the RLE compression by adding the prefix-sum oracle (\cref{dfn:prefix-sum}), which tells the position of a run in the uncompressed string. The addition of the prefix-sum oracle is necessary and efficient: finding an LCS from RLE inputs needs at least $\tOmg(n)$ queries due to a reduction of the $\PARITY$ problem (\cref{thm:LB-LCSRLE}); constructing the prefix-sum oracle takes $O(\tn)$ time and saving it takes $O(n)$ space, which are the same as those of RLE, so adding prefix sum oracle to RLE only incurs a constant overhead in the resource used.

To construct a quantum LCS algorithm between RLE strings, a major challenge is that the longest common substring in terms of encoded length may differ from the longest common substring in terms of decoded length. For example, between $\torle{abcdbbbbccccc}$ and $\torle{abcd@bbbbcc}$, the RLE string $\torle{abcd}$ is the longest one in terms of encoded length, and $\torle{bbbbcc}$ is the longest one in terms of decoded length, which is what we want to find.
Therefore, applying existing LCS algorithms for strings directly on RLE inputs will not work.

Our algorithm is nearly optimal, and we prove a matching lower bound  in \cref{thm:LB-LCSRLEP}.

\subsection{Related work}

Gibney, Jin, Kociumaka and Thankachan \cite{GibneyJKT24, GibneyT23} developed a $\tO(\sqrt{zn})$-time quantum algorithm for the Lempel-Ziv77 algorithm (LZ77) \cite{ZivL77} and calculating the Run-Length-encoded Burrows-Wheeler Transform (RL-BWT), where $n$ is the length of the input string and $z$ is the number of factor in the LZ77 factorization of the input, which roughly corresponds to the encoded length of that string. Given two strings $A$ and $B$, they showed how to calculate the LZ77 compression and a supporting data structure of $A\$B$. With this compressed data, the LCS between $A$ and $B$ can be found efficiently.

Note that the model \cite{GibneyJKT24, GibneyT23} is different from our model because they need do preprocessing on the \emph{concatenated string} $A\$B$, while in our work, we can preprocess and compress $A$ and $B$ independently, so we can compress and store the strings in the downtime, and when two compressed strings need to be compared, their LCS can be calculated in time \emph{almost independent of the uncompressed length}, which is potentially much faster than running the compression of \cite{GibneyJKT24, GibneyT23} on the uncompressed strings.

\subsection{Overview of the algorithm}



Our algorithm is built on and modified from the {LCS with threshold} algorithm of \cite{jin2023quantum}, which decides whether a common substring of length at least $d$ exists in time $\tO(n^{2/3}/d^{1/6-o(1)})$.

As stated in the introduction, the main obstacle we face is the difference between encoded length and decoded length. To overcome this difficulty, we perform ``binary searches'' on both encoded length and decoded length with a nested search. The outer loop is a binary search on the decoded length of the answer $\td\in[\tn]$. In each iteration of this binary search, we check whether a common substring of decode length at least $\td$ exists. The inner loop searches over encoded length $d=n/2, n/4, n/8, \dots$. In each iteration of the inner loop, we check whether a common substring with encoded length in $[d,2d]$ and decoded length at least $\td$ exists. Running these two loops gives us an $\Od(\log(n)\log(\tn))=\tO(1)$ overhead.

Note that there is a subtle issue in our inner loop search: unlike the original LCS problem, where having a common substring of length $d$ guarantees that there is a common substring of length $d-1$, in our LCS between RLE problem, there might be no common substring of encoded length $d-1$ and decoded length $\td$, even though there is a common substring of encoded length $d$ and decoded length $\td$. Therefore, in our inner-loop search, we need to search over every possible encoded string length. We accomplish this by modifying the algorithms of \cite{jin2023quantum} so that instead of just checking $d$, it will check the range $[d,2d]$ and loops over $d=n/2, n/4, n/8, \dots$.



Similarly to \cite{jin2023quantum}, in each iteration of our inner loop search, we run Ambainis' element distinctness algorithm \cite{Ambainis07} on an anchor set of $A\$B$, where $A$ and $B$ are the input RLE strings of length $n$. Roughly speaking, a $d$-anchor set of a concatenated string $A\$B$ is a subset of $A\$B$ such that if a common substring of length $d$, $A[i_1:i_1+d-1]=B[i_2:i_2+d-1]$, exists, the respective copies of the common substring in $A$ and $B$ will be ``anchored'' at the same positions, meaning that there exists a shift $0\leq h\leq d $ such that both $i_1+h$ and $n+1+i_2+h$ are in the anchor set. The $d$-anchored set has size roughly $n/d$. In each iteration of our inner loop search, we run a quantum walk on the elements of the $d$ anchored set and check for a ``collision'': a pair of anchored positions in $A$ and $B$ that can be extended backward and forward into a pair of common substrings with
encoded length in $[d,2d]$ and decoded length at least $\td$ and elements of it can be computed in time $\sqrt{d}$ by the construction in \cite{jin2023quantum}. To check both conditions in the encoded length and decoded length, we have a delicate checking procedure: to check everything with encoded length $[d, 2d]$, we search over all possible shifts of the anchor in \emph{encoded} length;
 To efficiently check for the decoded length, we store data about the lexicographically sorted \emph{decoded} prefixes and suffixes of the stored anchors in the data structure.

\section{Preliminaries}\label{sec:prelim}

\subsection{Conventions and Notations}
We abbreviate both ``run-length encoding'' and ``run-length-encoded'' to ``RLE''
We use tilde ( $\wt{\PH}$ ) to denote decoded strings and their properties, while notations without tilde refer to their RLE counterparts.
We use calligraphic letters (\eg $\cA$ and $\cB$) to denote algorithms, use teletype letters (\eg $\texttt{a}$) to denote strings or character literals, and use sans-serif letters (\eg $\LCS$) to denote problems.
We count indices from $1$.
By $[m]$, we mean the set $\set{1,2,\ldots,m}$.
The asymptotic notations $\tO(\PH)$ and $\tOmg(\PH)$ hide $\polylog(n)$ and $\polylog(\tn)$ factors, where $n$ is the encoded length of the input, and $\tn$ is the decoded length of the input. We say that a quantum algorithm succeeds with high probability if its success probability is at least $\Omg(1-1/\poly(n))$.

\subparagraph*{Strings.}
A \emph{string} $\ts\in\Sigma^\ast$ is a sequence of characters over a character set $\Sigma$.
The length of a string $\ts$ is denoted as $|\ts|$.
For a string $\ts$ of length $n$, a \emph{substring} of $\ts$ is defined as $\ts[i:j] := \ts[i':j'] = \ts[i']\ts[i'+1]\ldots\ts[j']$, where $i'=\max(1,i)$ and $j'=\min(n,j)$. \Ie it starts at the $i'$-th character and ends at the $j'$-th character.
If $i>j$, we define $\ts[i:j]$ as an empty string $\eps$. $\ts^R$ denotes the reverse of $\ts$: $\ts^R=\ts[n]\ts[n-1]\dots \ts[1]$.
%

 $\ts \prec \Tt$ denotes that $\ts$ is lexicographically smaller than $\Tt$. We use $,\preceq, \succ,\succeq$ analogously. 
\subsection{Run-Length Encoding}
\emph{Run-Length encoding} (RLE) of a string $\ts$, denoted as $s$, is a sequence of runs of identical characters $s[1]s[2]\cdots s[n]$, where $s[i]$ is a maximal run of identical characters, $n$ is the length of $s$, \ie the number of such runs.
For a run $s[i]$, $R(s[i])$ is its length and $C(s[i])$ denotes the unique character comprising the run.
When we write out $s$ explicitly, we write each $s[i]$ in the format of $C(s[i])^{R(s[i])}$, with $C(s[i])$ in a teletype font (\eg $\texttt{a}^3$).
Equivalently, each run $s[i]$ can be represented as a character-length pair $(C(s[i]), R(s[i]))$.
When there exist $i$ and $j\geq i$ such that $t = s[i:j]$, we call $t$ an \emph{substring} of $s$. In addition, we define \emph{generalized substring} of an RLE string as follows:

\begin{definition}[Generalized substring of an RLE String] \label{dfn:rle-substr}
For two RLE strings $s$ and $t$, we say $s$ is a \emph{generalized substring} of $t$ if $\ts$ is a substring of $\Tt$.
\end{definition}
For example, for $t=\torle{aaabbbbccddddd}$, the RLE string $\torle{bbbbcc}$ is a substring as well as a generalized substring, while $\torle{abbbbccdd}$ is a generalized substring but not a substring.

Our algorithm needs to know the location of a run of an RLE string in the original string. This is formalized by the ability to query an oracle of the following prefix-sum function:
\begin{definition}[prefix-sum of the runs of an RLE string]\label{dfn:prefix-sum}
For an RLE string $s$, $P_s[i]$ is the $i$th \emph{prefix-sum} of the runs, \ie $P_s[i] := \sum_{j=1}^iR(s[j])$, with $P_s[0] := 0$. Intuitively, $P_s[i]$ is the index where $s[i]$, the $i$-th run in $s$, ends in the decoded string $\ts$. As a consequence, for $i\leq j$, the decoded length of $s[i:j]$ is $P_s[j] - P_s[i-1]$.
\end{definition}

Note that an oracle of prefix-sum can be constructed and stored in QRAM in linear time while doing the RLE compression, thus constructing it only adds a constant factor to the preprocessing time. 

Also note that given prefix-sum oracle, the inverse of prefix-sum can be calculated in $\Od(\log n)$ time:
\begin{lemma}[Inverse Prefix-sum, $\Piv_S$]\label{thm:inv-prefix}
Given a prefix-sum oracle of an RLE string $S$ of encoded length $\Od(n)$, one can calculate the function $\Piv_S:[\tn]\rightarrow[n]$ that maps indices of $\TS$, the decoded string, to the corresponding ones of $S$ in $\Od(\log n)$ time.
\end{lemma}
\begin{proof}
    Let $\ti\in[\tn]$ be a decoded index. To find the corresponding index $i\in[n]$, we do a binary search over $[n]$ to find the $i\in[n]$ such that $P_S[i-1]<\ti\leq P_S[i]$. The process is correct since the prefix-sum is strictly increasing.
\end{proof}

Finally, we will often compute the length of longest decoded common prefix of two RLE string in our algorithm, so we formalize it as follows:
\begin{definition}[length of longest decoded common prefix (ldcp)]\label{dfn:ldcp}
	For two RLE string $s$, $t$, we define $\ldcp(s,t)=\max\{ j: \ts[1,\dots,j]=\Tt[1,\dots,j]\}$
\end{definition}

The following lemma follows from a well-known fact.

\begin{lemma}[e.g. \cite{kent2012demand} lemma 1]\label{lem:lexico}
	Given strings $s_1 \prec s_2 \prec\dots \prec s_n$, we have $\ldcp(s_1,s_n)=\min_{1\leq i\leq n-1}\ldcp(s_i,s_{i+1})$
\end{lemma}


   	 

\subsection{Computation Model}

\subparagraph*{Quantum Oracle.}

Let $S$ be an RLE string. In a quantum algorithm, we access an RLE string $S$ via querying the oracle $O_S$. More precisely,
\begin{equation}
    O_S:
    \ket{i}\ket{c}_{\text{char}}\ket{r}_{\text{run}}
    \mapsto
    \ket{i}\ket{c\oplus C(S[i])}_{\text{char}}\ket{r \oplus R(S[i])}_{\text{run}}
\end{equation}
is a unitary mapping for any $i\in[|S|]$, any $c\in\Sigma$, and any $r\in[\tn]$.
The corresponding prefix-sum $P_S$ (see Definition~\ref{dfn:prefix-sum}) can be accessed from the unitary mapping
\begin{equation}
    O_P:
    \ket{i}\ket{x}
    \mapsto
    \ket{i}\ket{x\oplus P_S[i]},
\end{equation}
for any $i\in\set{0}\cup [|S|]$ and any $x\in[\tn]$.

\subparagraph*{Word RAM model.}

We assume basic arithmetic and comparison operations between two bit strings of length $O(\log(\tn))$ and $O(\log n)$ both cost $O(1)$ quantum time.

\subsection{Definitions}\label{sec:defs}

\begin{definition}[$d$-anchor set (Definition 4.1, Theorem 4.2, and Theorem 1.1 of \cite{jin2023quantum}).]\label{dfn:anchor}
For a concatenated string $T=S_1\$S_2$ of length $n$, $X=\{X(1),X(2),\\\ldots,X(m)\} \subseteq [n]$ is a $d$-anchor set if either one of the following holds: 1. a common substring $S_1[i:i+d]=S_2[j:j+d]$ of length $d$ and a shift $h\in [n]$ exist such that $i+h \in X$ and $|S_1|+1+j+h\in X$. 2. $S_1$ and $S_2$ do not have a common substring of length $d$.

The construction of the anchor set $X$ can depend on the contents of $S_1$ and $S_2$. There exists a $d$-anchor set of size $m\leq n/d^{1-o(1)}$ whose entries $X(i)$ can be computed using $\tO(d^{1/2+o(1)})$ quantum time when $i\in[m]$ is given.    
\end{definition}

\begin{definition}[Longest Common Substring (LCS)]
A string $\ts$ is a \emph{longest common substring} (LCS) of strings $\tA$ and $\tB$ if it is a substring of both, and $|\ts|\geq|\Tt|$ for every common substring.
\end{definition}

\begin{definition}[LCS Problem on RLE Strings]\label{dfn:LCSRLE}
 Given oracle access to two RLE strings $A$ and $B$, find the longest common generalized substring $s$ of $A$ and $B$, \ie the RLE of an LCS $\ts$ between $\tA$ and $\tB$, and locate an instance of $s$ in each input. More precisely, find a tuple $(i_A, i_B, \ell)$ such that $|s|=\ell$, with an instance of $s$ starting within the run $A[i_A]$, and another instance of $s$ starting within the run $B[i_B]$. We denote this problem by $\LCSRLE$.
\end{definition}

\begin{definition}[Decoded Length of LCS on RLE Strings Problem ($\DLLCSRLE$)]\label{dfn:DLLCSRLE}
Given oracle access to two RLE strings $A$ and $B$,
calculate $|\ts|$ such that $\ts$ is an LCS of $\tA$ and $\tB$.
We denote this problem as $\DLLCSRLE$.
\end{definition}

\begin{definition}[Encoded Length of LCS on RLE Strings Problem]\label{dfn:ELLCSRLE}
Given oracle access to two RLE strings $A$ and $B$,
calculate $|s|$ such that $\ts$ is an LCS of $\tA$ and $\tB$.
We denote this problem as $\ELLCSRLE$.
\end{definition}

In \cref{sec:LB} we show a near-linear lower bound on query complexity for $\DLLCSRLE$ and a near-linear, $\Omg(n/\log^2n)$, lower bound for $\ELLCSRLE$.
Both problems are bounded by reductions from the parity problem.

\begin{definition}[Parity Problem]\label{dfn:PARITY}
    Given oracle access to a length-$n$ binary string $B\in\set{0,1}^n$, find $\bigoplus_{i=1}^n B_i$, the parity of $B$, where $\oplus$ is addition in $\mathbb{Z}_2$. We denote this problem as $\PARITY$.
\end{definition}

With a short reduction, we show that $\ELLCSRLE$ and $\LCSRLE$ share the same lower bound on query complexity (\cref{thm:LB-LCSRLE}).
As a result, we loosen the requirement and assume that the oracle of prefix-sum of the inputs is also given.



Our main algorithm solves the LCS problem with the prefix-sum oracle provided, formalized below.

\begin{definition}[LCS Problem on RLE Strings, with Prefix-sum Oracles]\label{dfn:LCSRLEP}
Given oracle access to two RLE strings $A$ and $B$ and prefix-sums of their runs, $P_A$ and $P_B$,
find an RLE string $s$, such that $\ts$ is an LCS of their decoded counterparts $\tA$ and $\tB$.
More precisely, the algorithm outputs the same triplet as the one for $\LCSRLE$.
We denote this problem as $\LCSRLEP$.
\end{definition}


\begin{definition}[Longest Repeated Substring problem on RLE string]
Given oracle access to an RLE string $A$ and prefix-sums of it runs, $P_A$,
find an RLE string $s$, such that $\ts$ is a longest repeated substring of the decoded string $\tA$.
More precisely, the algorithm outputs $(i_1,\ i_2, \ell)$ for two heads of the generalized substrings and its encoded length.

The longest repeated substring of a string $\ts$ of size $\tn$ is a string $\Tt = s[i:i+\ell-1] = s[j:j+\ell-1]$ for distinct $i,j\in[\tn]$ with the maximum possible $\ell$.
\end{definition}

\subsection{Primitives}\label{sec:primitive}

\subparagraph*{Grover's search (\cite{Grover96}).}
Let $f:[n]\rightarrow\set{0,1}$ be a function. There is a quantum algorithm $\cA$ that finds an element $x\in[n]$ such that $f(x)=1$ or verifies the absence of such an element.
$\cA$ succeeds with probability at least $2/3$ and has time complexity $\tO(\sqrt{n}\cdot T)$, where $T$ is the complexity of computing $f(i)$.


\subparagraph*{Amplitude amplification (\cite{BrassardH97}, \cite{Grover98}).}
Let $\cA$ be a quantum algorithm that solves a decision problem with one-sided error and success probability $p\in(0,1)$ in $T$ quantum time. There is another quantum algorithm $\cB$ that solves the same decision problem with one-sided error and success probability at least $2/3$ in $\tO(T/\sqrt{p})$ quantum time.

\subparagraph*{Minimum finding (\cite{Durr96}).}
Let $f:[n]\rightarrow X$ be a function, where $X$ is a set with a total order. There is a quantum algorithm $\cA$ that finds an index $i\in[n]$ such that $f(i)\leq f(j)$ for all $j\in[n]$. $\cA$ succeeds with probability at least $2/3$ and costs $\tO(\sqrt{n}\cdot T)$ time, where $T$ is the time to compare $f(i)$ to $f(j)$ for any $i,j\in[n]$.

\subparagraph*{Element distinctness (\cite{Ambainis07}, \cite{LeGallS22})}\hspace{-1em}\footnotemark{}
Let $X$ and $Y$ be two lists of size $n$ and $f:(X\cup Y)\rightarrow\mathbb{N}$ be a function. There is a quantum algorithm $\cA$ that finds an $x\in X$ and a $y\in Y$ such that $f(x)=f(y)$. $\cA$ succeeds with probability at least $2/3$ and costs $\tO(n^{2/3}\cdot T(n))$ time; $T(n)$ is the time to make the three-way comparison between $f(a)$ and $f(b)$ for any $a,b\in X\cup Y$.

\footnotetext{The definition here is also known as claw finding. The time upper bound is obtained in \cite[Section 2.1]{LeGallS22}. We also explain it in \cref{sec:QW}}

\begin{lemma}\label{thm:2d range}[2D range sum (Lemma 3.15 of \cite{AkmalJ22})]
    Let $K$ be a set of $r$ (possibly duplicated) points in $[n]\times[n]$. There exists a history-independent data structure of $K$ that, with $1\leq x_1\leq x_2\leq n$ and $1\leq y_1\leq y_2\leq n$ given, returns the number of points in $[x_1,x_2]\times[y_1,y_2]$ using $\tO(1)$ time. Also, entries can be inserted into and deleted from the data structure in $\tO(1)$ time.
\end{lemma}

\begin{lemma}[Dynamic array (Lemma 3.14 in \cite{jin2023quantum}).]\label{thm:array}
	\newcommand{\K}{\textsf{key}}
	\newcommand{\V}{\textsf{value}}
	There is a history-independent data structure of size $\tO(r)$ that maintains an array of key-value pairs $(\K_1,\V_1)\\ \ldots(\K_r,\V_r)$ with distinct keys and supports the following operations with worst-case $\tO(1)$ time complexity and high success probability:
	\begin{itemize}
    	\item \textbf{Indexing}:
          	Given an index $1\leq i \leq r$, return the $i$-th key-value pair.
    	\item \textbf{Insertion}:
          	Given an index $i\leq i \leq r+1$ and a new pair, insert it into the array between the $(i-1)$-th and the $i$-th pair and shift the later item to the right.
    	\item \textbf{Deletion}: Given an index $1\leq i \leq r$, delete the $i$-th pair from the array and shift later pairs to the left.
    	\item \textbf{Location}:
          	Given a key, return its index in the array.
    	\item \textbf{Range-minimum query}:
          	Given $1\leq a \leq b \leq r$, return $\min_{a\leq i\leq b}\{\V_i\}$.
	\end{itemize}
\end{lemma}

\begin{lemma}[Boost to high success probability]\label{thm:whp}
Let $\cA$ be a bounded-error quantum algorithm with time complexity $\Od(T)$.
By repeating $\cA$ for $\Od(\log n)$ times then outputting the majority of the outcomes, we can boost the success probability of $\cA$ to $\Omg(1-1/\poly(n))$ with overall time complexity $\Od(T\cdot\log n)$.
\end{lemma}

\cref{thm:whp} enables us to do Grover's search over the outcomes of applying $\cA$ on different inputs, because quantum computational errors accumulate linearly.\footnote{In fact, it is possible to apply Grover's search over bounded-error verifier \emph{without} the logarithmic overhead \cite{HoyerMW03}.}

\section{LCS from two RLE strings with Prefix-sum Oracles}\label{sec:meat}

\subsection{Quantum walk search}\label{sec:QW}

We use the quantum walk framework of \cite{MagniezNRS11} on a Johnson graph.

A Johnson graph, denoted as $J(m,r)$, where $r$ is a number to be chosen later, consists of $\binom{m}{r}$ vertices, each being an $r$-sized subset $R$ of a list $S$ of size $m$. 
In $J(m,r)$, two vertices $R_1$ and $R_2$ are connected iff $|R_1\cap R_2|=r-1$.


Associated with each vertex $R$, is a data structure $D(R)$ that supports three operations: setup, update, and checking; whose costs are denoted by $s(r)$, $u(r)$, and $c(r)$, respectively.
The setup operation initializes the $D(R)$ for any vertex $R$;
the update operation transforms $D(R)$ into $D(R')$ which is associated with a neighboring vertex $R'$ of $R$ in the graph;
and the checking operation checks whether the vertex $R$ is marked, where the meaning of marked will be defined later.
The MNRS quantum walk search algorithm can be summarized as:

\begin{theorem}[MNRS Quantum Walk search \cite{MagniezNRS11}]\label{thm:MNRS}
Assume the fraction of the marked vertices is zero or at least $\dta$.
Then there is a quantum algorithm that always rejects when no marked vertex exists; otherwise, with high probability, it finds a marked vertex $R$.
The algorithm has complexity
    \begin{equation}\label{eqn:MNRS}
        \tO\left(s(r) + \frac{1}{\sqrt{\dta}}\left(\sqrt{r}\cdot u(r)+c(r)\right)\right)
    .\end{equation}
\end{theorem}

\begin{remark}
    As noted in \cite{Ambainis07}, the data structure associated with data of each vertex of the quantum walk needs to be history-independent. I.e. the form of some data stored in the data structure is independent of history of insertions and deletions to aggregate these data.
\end{remark}

\subsection{The algorithm}\label{sec:alg}


\begin{theorem}[Algorithm for $\LCSRLEP$]\label{thm:algo}

    Given oracle access to RLE strings $A$ and $B$, and their prefix-sums, there exists a quantum algorithm $\cA$ that, with high probability, finds a 3-tuple $(i_A,i_B,|s|)$ that identifies a longest common generalized substring (see \cref{dfn:rle-substr}) $s$ between $A$ and $B$ if it exists; otherwise $\calA$ rejects. $\cA$ has a time cost $\tO(n^{2/3}/d^{1/6-o(1)})\cdot\Od(\log\tn)$, where $n$ and $\tn$ are the encoded length and the decoded length of input strings, respectively.
\end{theorem}
\begin{proof}
	We give a constructive algorithm here. The high level structure of the algorithm is summerized as Algorithm~\ref{alg:outer_loops}.

	The algorithm runs two loops. The outer loop binary searches over $\td\in[\tn]$. In each iteration of this binary search, we check whether a common substring of decode length at least $\td$ exists. The inner loop searches over encoded length $d=n/2, n/4, n/8, \dots$. In each iteration of the inner loop, we check whether a common substring with encoded length in $[d,2d]$ and decoded length at least $\td$ exists. Running these two loops gives us an $\Od(\log(n)\log(\tn))=\tO(1)$ overhead. 

    \SetKwFor{BinarySearch}{BinarySearch}{such that {\it  search\_flag=1}}{}
    \begin{algorithm2e}\label{alg:outer_loops}
      \caption{Algorithm for $\LCSRLEP$}
        \KwIn{RLE strings $A$, $B$, their decoded length $\tn$, encoded length $n$, and prefix-sums $P_A$ and $P_B$.}
        \KwOut{$(i_A,i_B,|s|)$ associated with the longest common generalized substring of $A$ and $B$ }
        \BinarySearch{for maximal $\td\in[\tn]$ }{
        search\_flag $\gets$ 0 \\
            \For{encoded length $d\in\{n/2,n/4,n/8,\ldots\}$}{
                $(i_A,i_B,|s|)$ $\gets$ Quantum Walk search on anchor set $X$ of size $m=n/d^{1-o(1)}$ with $ r=O(m^{2/3})$ for anchors $i_A$, $i_B$ of a common generalized substring $s$ with $|\ts|\geq \td$ and $|s|\in[d, 2d]$.\label{line:qw}  \\
                \If{{line}~\ref{line:qw} found a marked item}{search\_flag $\gets$ 1 \\
                Record ($i_A$, $i_B$, $|s|$)\\
                \textbf{break}
                }
            }
        }
         \Return ($i_A$, $i_B$, $|s|$)
    \end{algorithm2e}

	Let $X$ be the $d$-anchor set on the concatenated RLE string $S=A\$^1B$. As stated in {\cref{dfn:anchor}}, $X$ has size $m =n/d^{1-o(1)}$. Let $\tS$ be the decoded string of $S$. For an index $k\in[m]$, we define the following decoded ``prefix'' and ``suffix'' strings of encoded length $2d$:
	\begin{align}
    	P(k)&= \tS[P_S(X(k)-1)+1 :P_S(X(k)+2d)]\\
    	Q(k)&= \tS[P_S(X(k)-2d-1)+1:P_S(X(k))]^R,
	\end{align}
	where the prefix sum oracle $P_S[{}\cdot{}]$ is defined in \cref{dfn:prefix-sum}.
    
	To check whether a common substring with encoded length in $[d,2d]$ and decoded length at least $\td$ exists, we run the MNRS quantum walk of \cref{thm:MNRS} on the Johnson graph $J(m,r)$, where each vertex represents a subset of $r$ items out the $m$ items in the anchor set $X$. For each vertex on the Johnson graph, we store the following data in the associated data structures:
	     \newcommand{\kplist}{(k_1^P, k_2^P,..., k_r^P)}
        \newcommand{\hplist}{(h_1^p,h_2^P,\dots,h_{r-1}^P)}
        \newcommand{\kqlist}{(k_1^Q, k_2^Q,..., k_r^Q)}
        \newcommand{\klist}{(k_1,k_2,\dots,k_r) }
 \begin{enumerate} 
    	\item Indices (in the anchor sets) of the $r$ chosen points sorted according to their values: $(k_1,k_2,\dots,k_r)  \in [m]^r$ such that $k_i <k_{i+1}$ for all $i$. 
 
    	\item corresponding positions of the chosen anchors on the encoded string:
   	 
    	$X(k_1),\dots,X(k_r) \in [|S|]$.

    	\item The indices $(k_1,k_2,\dots,k_r )$ sorted according to the decoded string after them: An array $\kplist$, which is a permutation of $(k_1,k_2,\dots,k_r )$, satisfying that $P(k_i^P) \preceq P(k_{i+1}^P)$ for all $i$.
   
     \item The array of length of LCP between $k_i^P$: $\hplist$ where $h^P_i= \ldcp(P(k_i^P),P(k_{i+1}^P))$.\footnote{Recall that $\ldcp$ is defined in \cref{dfn:ldcp}.}
	\item The indices $(k_1,k_2,\dots,k_r )$ sorted according to the decoded string before them: An array $\kqlist$, which is a permutation of $(k_1,k_2,\dots,k_r )$, satisfying that $Q(k_i^Q) \preceq Q(k_{i+1}^Q)$ for all $i$.
    	\item The array of length of LCP between $k_i^Q$:  $(h_1^Q,h_2^Q,\dots,h_{r-1}^Q)$ where $h^Q_i= \ldcp(Q(k_i^Q),Q(k_{i+1}^Q))$
     \item Additional data to apply \cref{thm:2d}
	\end{enumerate}
We store $((k_1,X(k_1)),(k_2,X(k_2)),\dots, (k_r,X(k_r)))$, $\kplist$,  \\ $\hplist$, $\kqlist$, and $(h_1^Q,h_2^Q,\dots,h_{r-1}^Q)$
 in 5 different dynamic arrays of \cref{thm:array}.

For every index $k\in[m]$, we assign a color to it to specify whether the anchor is on string $A$ or string $B$: if $X(k)\leq n$, we say $k$ is red. If $X(k) \geq n+2$, we say $k$ is blue. if $X(k)=n+1$, we say $k$ is white. Since we store $X(k_1),\dots,X(k_r)$, we can look up the color of each $k_i$ in $\tO(1)$ time.

For every $i \in [r]$, we define $\pos^P(i)$ as the index $j$ such that $k_i=k^P_j$ . Similarly,  define $\pos^Q(i)$ as the index $j$ such that $k_i=k^Q_j$. Finding $\pos^P(i)$  can be done in $\tO(1)$ with the indexing operation of $((k_1,X(k_1)),(k_2,X(k_2)),\dots, (k_r,X(k_r)))$ followed by the location operation of $\kplist$, and similarly for $\pos^Q(i)$.

Recall that the quantum walk is composed of the setup, update, and checking operations. The setup operation can be done by inserting $r$ elements of the anchor set into the stored data set. The update operation can be done by inserting an element and deleting an element. Since deletion can be done by reversing an insertion, both the setup and update operations can be done by multiple applications of the insertion procedure. This insertion procedure is summarized in Algorithm~\ref{alg:insertion}. 
The checking operation is summarized in Algorithm~\ref{alg-check}


\begin{algorithm2e}\label{alg:insertion}
\caption{The insertion procedure}
\KwIn{$k\in[m]$ to be inserted}
Compute $X(k)$\\
Compute $i$ such that $k_i\leq k < k_{i+1}$\\
Update $((k_1, X(k_1)), \ldots, (k_r, X(k_r)))$ $\gets$ $((k_1, X(k_1)), \ldots, (k_i, X(k_i)), (k, X(k)), (k_{i+1}, X(k_{i+1})), \ldots , (k_r, X(k_r)))$\\
Compute $j$ such that $P(k_{j}^P)\leq P(k) < P(k_{j+1})$ \\
Compute $h_p=\ldcp(P(k_j^P),P(k))$ \\
Compute $h_s=\ldcp(P(k_{j+1}^P),P(k))$ \\
Compute $h_o=\ldcp(P(k_{j}^P), P(k_{j+1}^P))$ \\
Update $(k_1^P,\ldots,k_r^P) \gets (k_1^P,\ldots,k_j^P,k,k_{j+1}^P,\ldots,k_r^P)$\\
Update $(h_1^P,\ldots,h_r^P) \gets (h_1^P,\ldots,h_{j-1}^P,h_p,h_s,h_{j+1}^P,\ldots,h_r^P)$\\
Compute $j, h_p,h_s,h_o$ for $Q$\\
Update $(k_1^Q,\ldots,k_r^Q) \gets (k_1^Q,\ldots,k_j^Q,k,k_{j+1}^Q,\ldots,k_r^Q)$\\
Update $(h_1^Q,\ldots,h_r^Q) \gets (h_1^Q,\ldots,h_{j-1}^Q,h_p,h_s,h_{j+1}^Q,\ldots,h_r^Q)$\\
\end{algorithm2e} 

We now outline the steps involved in Algorithm~\ref{alg:insertion}
 and analyze its time complexity. To insert an entry $k\in [m]$ into the data structure, Algorithm~\ref{alg:insertion} follows these steps:
\begin{enumerate}
	\item Compute $X(k)$ in time $d^{1/2+o(1)}$.
\item Find $i$ such that $k_i\leq k < k_{i+1}$ in the ordered array $\klist$ by binary search in time $\tO(1)$. Insert $(k,X(k))$ into the $i$-th position of $((k_1,X(k_1)),(k_2,X(k_2)),  \dots, (k_r,X(k_r)))$ in time $\tO(1)$. 
\item Find the position $i$ to insert $k$ in the ordered array $(k_1^P,... k_r^P)$ by binary search. In each iteration of the binary search, we need to compare the lexicographical order between $P(k)$ and $P(k_j^P)$ for some $j$. Since the strings have compressed length $O(d)$, the comparison can be done in time $O(\sqrt{d})$ by finding the first run where they are different through minimum finding. This step can be done in time $\tO(\sqrt{d})$.
\item  Use minimum finding and the prefix sum oracle to compute $h_p=\\ \ldcp(P(k_i^P),P(k))$, $h_s=\ldcp(P(k_{i+1}^P),P(k))$, and $h_o=\ldcp(P(k_i^P),P(k_{i+1}^P))$ in $O(\sqrt{d})$ time. Update $(h_1^p,\dots,h_{r-1}^P)$ by inserting $h_p,h_s$ and uncompute $h_o$.
\item Do the same for $Q$ in time $\tO(\sqrt{d})$.
\end{enumerate}
Therefore the insertion cost is $\tO(d^{1/2+o(1)})$. To delete an element, we can reverse the insertion in $\tO(d^{1/2+o(1)})$ time.

\SetKwFor{GroverSearch}{GroverSearch}{}{}
\begin{algorithm2e}\label{alg-check}
\caption{The checking procedure}
\GroverSearch{ $d'\in [0,2d]$ and $r' \in [r]$}{
    \eIf{$k_{r'}$ is {\it red}}
    {$\mathsf{flag\_color}\gets$ blue}
    {$\mathsf{flag\_color}\gets$ red} 
    $L\gets P_S(X(k_{r'}))-P_S(X(k_{r'})-d'-1)$ \\
    Find $l^Q,r^Q$ such that $\ldcp(Q(k_i^Q),Q(k_{r'}))\geq L$ if and only if $l^Q\leq i\leq r^Q$. \\
    Find $l^P,r^P$ such that $\ldcp(P(k_i^P),P(k_{r'}))\geq \td-L$ if and only if $l^P\leq i\leq r^P$. \\
    Check the existence of a $j' \in [r]$ such that the color of $k_{j'}$ is $\mathsf{flag\_color}$, $l^Q \leq \pos^Q({j'}) \leq r^Q $, and $l^P \leq \pos^P({j'}) \leq r^P$.  \\
    \If{$j'$ found}{\Return marked}
} 
\Return unmarked 
\end{algorithm2e} 



Next, we outline the steps and analyze the time cost of Algorithm~\ref{alg-check}. To check whether the subset of $r$ points is marked, we perform a Grover search over $d' \in [0, 2d]$ and $r' \in [r]$ to determine if the $r'$-th stored item anchors a common substring $s$ such that $r'$ is roughly at the $d'$-th run of $s$ and $\ts \geq \td$. The checking for each $(d', r')$ can be done by the following sub-algorithm in $\tO(1)$ time:
\begin{enumerate}
	\item If $k_{r'}$ is red, set $\mathsf{flag\_color}$ = blue. If $k_{r'}$ is blue, set $\mathsf{flag\_color}$ = red.
	\item Compute $L=L(d')=P_S(X(k_{r'}))-P_S(X(k_{r'})-d'-1)$.
	\item Find $l^Q,r^Q$ such that $\ldcp(Q(k_i^Q),Q(k_{r'}))\geq L$ if and only if $l^Q\leq i\leq r^Q$. To find $l^Q$, we calculate $j'=\pos^Q(r')$ and do a binary search to find the minimum $l\in [j']$ such that the range minimum of $(h^Q_l, \dots ,h^Q_{j'})$ is greater or equal to $L$. The range minimum query can be done in $\tO(1)$ time by \cref{thm:array}. This guarantees that $\ldcp(Q(k_i^Q),Q(k_{r'}))\geq L$ for $l^Q\leq i \leq j'$ because by \cref{lem:lexico}, $\ldcp(Q(k_i^Q),Q(k_{r'}))$ is equal to the range minimum
	of $(h^Q_i, \dots ,h^Q_{j'})$. Also , $l^Q$ always exists because $h^Q_{j'} \geq L$.  Similarly, to find $r^Q$ we do a binary search to find the maximum $w\in [j':r]$ such that the range minimum of $(h^Q_{j'}, \dots ,h^Q_{w})$ is greater or equal to $L$.
	\item Find $l^P,r^P$ such that $\ldcp(P(k_i^P),P(k_{r'}))\geq \td-L$ if and only if $l^P\leq i\leq r^P$.
	\item\label{line:2d} Check the existence of a $j' \in [r]$ such that the color of $k_{j'}$ is $\mathsf{flag\_color}$, $l^Q \leq \pos^Q({j'}) \leq r^Q $, and  $l^P \leq \pos^P({j'}) \leq r^P$. If such ${j'}$ exist, return marked. This can be done in $\tO(1)$ time by \cref{thm:2d}.
\end{enumerate}
 If the Grover search does not find a marked item, return unmarked.
 
The checking cost is $\tO(\sqrt{rd})$ since we are Grover searching over $2dr$ items.

Note that the first run of an RLE-compressed common substring does not necessary equal to the corresponding runs of both input strings because one of them can be longer. If a common substring whose length of the first run matches the length of the corresponding run of the red string, and encoded length in $[d,2d]$ and decoded length $\td$ exist, it will be anchored by the $d$ anchor set with some shift $h$, so when $d’=h$ we will find a collision with $k_{r'}$ being red. Similarly for common substrings whose length of the first run matches the length of the corresponding run of the blue string, we will find a collision with $k_{r'}$ being blue.

Finally we summarize the costs of the quantum walk. The setup can be done by $r$ insertions, so the cost is $s(r)=O(r d^{1/2+o(1)})$. The update is done by an insertion and a deletion, so the update cost is $u(r)=O(d^{1/2+o(1)})$.
The checking cost is $\tO(\sqrt{rd})$.

The fraction of marked vertices $\dta$ is lower bounded by $\Omg(r^2/m^2)$ since in the worst case there is only one mark vertex, and thus $\binom{m-1}{r-1}^2$ out of $\binom{m}{r}^2$ pairs of subsets are marked.

Putting things together, by choosing $r=O(m^{2/3})=\tO(n^{2/3}/d^{1/3-o(1)})$, the time complexity of the algorithm is

\[
	\tO\left(rd^{1/2+o(1)} + \sqrt{\frac{m^2}{r^2}} \left(\sqrt{r} \cdot d^{1/2+o(1)} + \sqrt{rd} \right) \right)=\tO(n^{2/3}/d^{1/6-o(1)})
\]

After finding a marked vertex, we do another checking operation on the marked subset to find the anchors $X(k_{r'})$, $X(k_{j'})$, and the shift $d'$. Then we compute $i_A=X(k_{r'})-d'$, $i_B=X(k_{j'})-d'$, and $\ell=P_A^{-1}(P_A[i_A]+\td)$, where $P_A^{-1}$ is the inverse of the prefix sum oracle\footnote{$P_A^{-1}$ can be computed in $\tO(1)$ time by binary search.}.
Finally, we output $(i_A,i_B,\ell)$.
\end{proof}

\begin{lemma}[(Algorithm 3 and 4 in \cite{jin2023quantum}).]\label{thm:2d}
  \cref{line:2d} of the checking procedure of the quantum walk can be implemented in $\tO(1)$ time with some modifications to the algorithm.

\end{lemma}
\begin{proof}
  We apply the following modifications.
  
  Before the quantum walk, we sample an $r$-subset $V$ of $[m]$, and denote  the ranking of $k^P_i$ among $\{P(v)|v\in V\}$ as $\rho^P(k^P_i)$. Define $\rho^Q({}\cdot{})$ similarly. We store $V$ in lexicographical order. This requires $O(r\sqrt{d})$ time, which is on the same order of the setup cost.
  
  Without loss of generality, let $\mathsf{flag\_color}$ be blue. To check whether a blue $k_{j'}$ exists,
  we maintain a dynamic 2D range sum data structure (\cref{thm:2d range}) to store $(\rho^P(k), \rho^Q(k))$ for blue $k$.
  In the checking operation, after finding $l^P$, $r^P$, $l^Q$, and $r^Q$, we  check whether the range $[\rho^P(k_{l^P})+1\ldots \rho^P(k_{r^P})-1] \times [\rho^Q(k_{l^Q})+1\ldots \rho^Q(k_{r^Q})-1]$ is non-zero in the dynamic 2D range sum data structure.
  If so, such a $k_{j'}$ exists.
  Otherwise, we check at most $O(\log m)$ blue $k$ with $\rho^P(k)\in\{\rho^P(k_{l^P}), \rho^P(k_{r^P})\}$ or $\rho^Q(k)\in\{\rho^Q(k_{l^Q}), \rho^Q(k_{r^Q})\}$ explicitly whether the lexicographical ranking of its prefix and suffix are in $[l^P,r^P]$ and in $[l^Q,r^Q]$, respectively.
  The insertion and deletion to the 2D range sum data structure can be done in $\tO(1)$ time.
  The checking algorithm uses $\tO(1)$ time and has $1/poly(m)$ one-sided error.
  
\end{proof}


\begin{remark}
	By dropping the red-blue constraint in the checking step, and receiving only one input string $S$ of encoded length $n$, we can adapt \cref{thm:algo} to solve the Longest Repeated Substring problem.
\end{remark}

\newcommand{\ta}{\texttt{a}}
\newcommand{\tb}{\texttt{b}}
\newcommand{\tc}{\texttt{c}}
\newcommand{\tat}{\texttt{@}}
\newcommand{\tsr}{\texttt{\#}}
\section{Lower Bounds}\label{sec:LB}

In this section, we first show a query lower bound for $\LCSRLEP$.
We then investigate the time complexity lower bound for calculating the length (encoded and decoded) of longest common substring from two RLE strings without access to prefix-sum oracles. And use the results to show the lower bound of finding an LCS from two RLE strings is $\tOmg(n)$, which is our motivation and justification to introduce the prefix-sum oracle.

\subsection{Lower Bound on \texorpdfstring{$\LCSRLEP$}{LCSRLEP}}\label{sec:LB-LCSRLEP}
\begin{lemma}[Lower Bound of $\LCSRLEP$]\label{thm:LB-LCSRLEP}
    Any quantum oracle algorithm $\calA$ requires at least $\tOmg(n^{2/3}/d^{1/6})$ queries to solve $\LCSRLEP$, with probability at least $2/3$.
\end{lemma}

\begin{proof}
	Note that for every string $A=a_1 a_2,\dots $, we can insert a special character $@$ to create a string $A_@=a_1 @ a_2 @ \dots$, so that $A_@$ cannot be RLE compressed. Also note that the length of LCS between any pair $(A_@,B_@)$ is exactly twice the length of  LCS between $(A,B)$, so any lower bound on LCS is also a lower bound on $\LCSRLEP$.  Prefix-sum oracle makes no difference since it can be calculated by multiplying the index of the runs by 2. Therefore, this lower bound follows Theorem 1.2 of \cite{jin2023quantum}.
\end{proof}


\subsection{Lower Bound on \texorpdfstring{$\DLLCSRLE$}{DLLCSRLE}}\label{sec:LB-DLLCSRLE}

In this section we show how to reduce $\PARITY$ to $\DLLCSRLE$, obtaining the following result.

\begin{lemma}[Lower Bound of $\DLLCSRLE$]\label{thm:LB-DLLCSRLE}
Any quantum oracle algorithm $\calA$ requires at least $\Omg(n)$ queries to solve $\DLLCSRLE$, with probability at least $2/3$.
\end{lemma}

The main idea is to encode an $n$-bit binary string $B=B_1B_2\ldots B_n$ as an RLE string $S_B$,
in which $R(S_B[i]) = 2+B_i$.
Then, using $\calA$, we find the length of LCS of $S_B$ with itself.
Here, what $\calA$ outputs is basically the decode length of $S_B$, \ie $|\wtilde{S_B}|$.
From that, we can calculate the parity of $B$ easily,
and thus \cref{thm:LB-DLLCSRLE} is proven.

\begin{proof}
Given an $n$-bit binary string $B$, we can construct an RLE string $S_B$ as:
\begin{equation}\label{eqn:SB}
	S_B
        = \rle{a,B_1+2,b,B_2+2,a,B_3+2,b,B_4+2}\cdots\gma^{B_n+2},
\end{equation}
where $\gma$ is $\ta$ if $n$ is odd, otherwise it is $\tb$.

Then we assume the algorithm $\calA$ exists.
With $S_B$ as the inputs, the output of $\calA$, the decoded length of LCS between $S_B$ and itself, is
\begin{equation}
	\calA(S_B,S_B)
	= |\tS_B|
	= \sum_{i=1}^n (B_i+2)
	= 2n+\sum_{i=1}^n Bi,
\end{equation}
which has the same parity as $\bigoplus_{i\in[n]} B_i$, \ie the parity of $B$.
Therefore, by checking the lowest bit of $\calA(S_B,S_B)$, we can solve $\PARITY$ with no extra query.
Since solving $\PARITY$ requires $\Omg(n)$ queries, solving $\DLLCSRLE$ needs at least the same number of queries.
\end{proof}

\bibliography{bib.bib}

\appendix
\section{Lower Bounds on \texorpdfstring{$\ELLCSRLE$}{ELLCSRLE} and \texorpdfstring{$\LCSRLE$}{LCSRLE}}\label{sec:LB-LCSRLE}

In this section, we show a $\tOmg(n)$ lower bound on both $\ELLCSRLE$ (\cref{thm:LB-ELLCSRLE}) and $\LCSRLE$ (\cref{thm:LB-LCSRLE}).
More precisely, we reduce $\PARITY$ to $\ELLCSRLE$, which is then reduced to $\LCSRLE$.

Here is a high-level overview of the first reduction (from $\PARITY$ to $\ELLCSRLE$).
We encode an $n$-bit binary string $B=B_1B_2\cdots B_n$ into an RLE string $S_B$,
in a way similar to \cref{eqn:SB} in the proof of \cref{thm:LB-DLLCSRLE}.
We then assume an algorithm $\calA$ of query complexity $Q(\calA)$ for $\ELLCSRLE$ exists.
Using $\calA$, we construct an algorithm to compare the decoded length of $S_B$, \ie $|\tS_B|$, with any $k>0$.
We then use binary search on $k$ to find $|\tS_B|$, invoking $\calA$ for $\Od(\log n)$ times.
From $|\tS_B|$, we calculate the parity of $B$ without extra query.
Finally, since $\PARITY$ has query lower bound $\Omg(n)$, $Q(\calA)$ is at least $\tOmg(n)$, getting \cref{thm:LB-ELLCSRLE} below.

\begin{lemma}[Lower Bound of $\ELLCSRLE$]\label{thm:LB-ELLCSRLE}
    Any quantum oracle algorithm $\calA$ requires at least $\tOmg(n)$ queries to solve $\ELLCSRLE$, with probability at least $2/3$.
\end{lemma}

\begin{proof}
Given an $n$-bit binary string $B=B_1B_2B_3\ldots B_n\in\set{0,1}^n$,
we can construct an RLE string
\begin{equation}
    S_B
    = \ta^{2B_1+2}\tb^{2B_2+2}\ta^{2B_3+2}\tb^{2B_4+2}\ldots\gma^{2B_n+2}
,\end{equation}
where $\gma$ is $\ta$ if $n$ is odd, otherwise it is $\tb$.

For every positive natural number $k$, we can also construct an RLE string, simply by repeating another character: $ S_k = \tc^k$.
We then concatenate $S_B$ and $S_k$ together with different characters in the middle, getting
\newcommand{\Sat}[1][k]{S_{B,\tat,{#1}}}
\newcommand{\Ssh}[1][k]{S_{B,\tsr,{#1}}}
\newcommand{\tSat}[1][k]{\wtilde{S}_{B,\tat,{#1}}}
\newcommand{\tSsh}[1][k]{\wtilde{S}_{B,\tsr,{#1}}}
\begin{equation}\label{eqn:StatStsr}
      \Sat := S_B\tat^1S_k
    \quad\text{and}\quad
    \Ssh := S_B\tsr^1S_k.
\end{equation}

Let us check what we know about the LCS $\ts$ between $\tSat$ and $\tSsh$.
Firstly, $\tat$ and $\tsr$ are not in $\ts$ since none of them appears in $\tSat$ and $\tSsh$ at the same time.
Secondly, $\ts$ is a substring of $\tS_B$ or $\tS_k$, but not both,
because the character set of $\tS_B$, $\set{\ta,\tb}$, and the one of $\tS_k$, $\set{\tc}$, do not intersect.
Finally, $\ts$ is the ``longest'' common substring, so it is the longest one among $\tS_B$ and $\tS_k$.

Now we assume the algorithm $\calA$ in \cref{thm:LB-ELLCSRLE} exists,
and it has query complexity $Q(\calA)$.

Additionally, the success probability of $\calA$ can be boosted from constant to high probability with an extra logarithmic factor on its query complexity (\cref{thm:whp}).

With $\Sat$ and $\Ssh$ as inputs, $\calA$ outputs
\begin{align}
	\calA(\Sat,\Ssh)
	&=
	\begin{cases}
  	  |S_k|,& |\tS_k| > |\tS_B| \\
  	  |S_B| \text{ or } |S_k|,& |\tS_k| = |\tS_B|\\
  	  |S_B|,& |\tS_k| < |\tS_B|
	\end{cases}\\
	&=
	\begin{cases}\label{eqn:calA_B}
  	  1,& k > |\tS_B|\\
  	  n \text{ or } 1,& k = |\tS_B|\\
  	  n,& k < |\tS_B|
	\end{cases}
.\end{align}
 
For a given $B$, we use $\calA_B({}\cdot{})$ as a shorthand for $\calA(\Sat[\,\cdot\ ],\Ssh[\,\cdot\ ])$ in the following text.
Note that when $k = |\tS_B|$, two answers ($n$ and $1$) are possible,
and we only assume $\calA$ outputs one of them.
Thus, $\calA_B(k)$ is non-deterministic when $k = |\tS_B|$.
We will resolve this issue with a property of binary search later.

To find $|\tS_B|$, we do a binary search on $k$ to find a $k'\in[2n, 4n]$, such that $\calA_B(k'-1)=n$ and $\calA_B(k')=1$.%
\footnote{The search range $[2n, 4n]$ comes from $2n \leq |\tS_B|=\sum_i (2B_i+2) \leq 4n$.}
In the binary search, $\calA_B$ will not be called with the same $k$ twice so it does not matter whether $\calA_B$ and the underlying $\calA$ are deterministic or not.
So from now on, we treat $\calA_B$ as if it were deterministic.

Since there are two possible outputs for $\calA_B(k)$ when $k=|\tS_B|$
(the middle case in \cref{eqn:calA_B}),
each corresponds to a different result $k'$ for the binary search.
If $\calA_B(|\tS_B|)$ outputs $1$, we will get $k' = |\tS_B|$, the desired result.
But if $\calA_B(|\tS_B|)$ outputs $n$, we will get $k' = |\tS_B|+1$ instead.
We can detect if the latter one is the case from the parity of $k'$ because $ |\tS_B| = 2\sum_{i=1}^n (B_i+1) $ is always even, and thus we can correct the result accordingly.

With $|\tS_B|$ in hand, we then check if
\begin{equation}
    \sum_{i=1}^n B_i
    =
    \frac12\sum_{i=1}^n \left(2B_i + 2\right) - n
    =
    \frac12|\tS_B| - n
\end{equation}
is odd or even to determine the parity of $B$.

Alternatively, we can XOR the lowest bit of $n$ with the second-lowest bit of $k'$, which is the same as the one of $|\tS_B|$, directly.
Then the result is the parity of $B$.
This allows us to avoid correcting $k'$ explicitly.

In total, we use $Q(\calA)\log^2n$ queries to solve $\PARITY$.
The logarithmic factors come from boosting $\calA$ to high probability and the binary search.
Finally, solving $\PARITY$ requires $\Omega(n)$ queries so we have
\begin{equation}
	Q(\calA)\log^2n \in \Omg(n)
	\implies
	Q(\calA) \in \Omg(n/\log^2n) \in \tOmg(n),
\end{equation}
and \cref{thm:LB-ELLCSRLE} follows.
\end{proof}

Furthermore, an algorithm solving $\LCSRLE$ outputs a triplet $(i_A,i_B,\ell)$, where $\ell$ is the encoded length of LCS between the inputs, which is also the answer to $\ELLCSRLE$.
\Ie $\ELLCSRLE$ can be reduced to $\LCSRLE$ with no extra query to the input strings.
As a result, $\LCSRLE$ shares the same query lower bound with $\ELLCSRLE$.
This gives the corollary below.

\begin{corollary}[Lower Bound of $\LCSRLE$]\label{thm:LB-LCSRLE}
Any quantum oracle algorithm $\calA$ requires at least $\tOmg(n)$ queries to solve $\LCSRLE$, with probability at least $2/3$.
\end{corollary}

\cref{thm:LB-LCSRLE} is our motivation and justification to introduce the prefix-sum oracles.

\end{document}